\documentclass{article}
\usepackage{iclr2019_conference,times}
\usepackage[smallerops]{newtxmath}

\usepackage{hyperref}
\usepackage{url}

\usepackage{amsthm,amsmath,amssymb}
    \newtheorem{lemma}{Lemma}
    \newtheorem{theorem}{Theorem}
\usepackage{dirtytalk}
\usepackage{graphicx}
\usepackage{microtype}
\usepackage{nicefrac}
\usepackage{subcaption}
\usepackage{todonotes}
\newcommand{\term}[1]{\textbf{#1}}

\newcommand{\model}[1]{\ensuremath{\mathbb{M}_{\text{#1}}}}
\newcommand{\modelone}[1]{\ensuremath{\mathbb{M}^1_{\text{#1}}}}
\newcommand{\true}{\ensuremath{\mathcal{T}}}
\newcommand{\pred}{\ensuremath{\mathcal{C}}}
\newcommand{\graph}{\ensuremath{G}}
\newcommand{\loss}{\ensuremath{\mathcal{L}}}
\DeclareMathOperator{\AMI}{AMI}
\DeclareMathOperator{\cNMI}{cNMI}
\DeclareMathOperator{\NMI}{NMI}
\DeclareMathOperator{\Expect}{\mathbb{E}}

\DeclareMathAlphabet{\mathcal}{OMS}{cmsy}{m}{n}
\DeclareMathAlphabet{\mathbb}{U}{msb}{m}{n}

\makeatletter  
\def\mathcolor#1#{\@mathcolor{#1}}
\def\@mathcolor#1#2#3{%
  \protect\leavevmode
  \begingroup
    \color#1{#2}#3%
  \endgroup
}
\makeatother

\iclrfinalcopy 

\makeatletter  
\def\mathcolor#1#{\@mathcolor{#1}}
\def\@mathcolor#1#2#3{%
  \protect\leavevmode
  \begingroup
    \color#1{#2}#3%
  \endgroup
}
\makeatother

\begin{document}
\title{An exact No Free Lunch theorem for \\ community detection}
\author{Arya D. McCarthy, Tongfei Chen, \& Seth Ebner \\
Johns Hopkins University\\
}
\date{}
\maketitle

\begin{abstract} 
%
%
\small
A precondition for a No Free Lunch theorem is evaluation with a loss function which does not assume \emph{a priori} superiority of some outputs over others. 
A previous result for community detection by \citet{peel2017ground} relies on a mismatch between the loss function and the problem domain. The loss function computes an expectation over only a subset of the universe of possible outputs; thus, it is only \emph{asymptotically} appropriate with respect to the problem size.
By using the correct random model for the problem domain, we provide a stronger, exact No Free Lunch theorem for community detection. The claim generalizes to other set-partitioning tasks including core--periphery separation, \(k\)-clustering, and graph partitioning.
Finally, we review the literature of proposed evaluation functions and identify functions which (perhaps with slight modifications) are compatible with an exact No Free Lunch theorem. 

\end{abstract}

\section{Introduction}

A myriad of tasks in machine learning and network science involve discovering structure in data. Especially as we process graphs with millions of nodes, analysis of individual nodes is untenable, while global properties of the graph ignore local details. It becomes critical to find an intermediate level of complexity, whether it be communities, cores and peripheries, or other structures. Points in metric space and nodes of graphs can be clustered, and hubs identified, using algorithms from network science. 
A longstanding theoretical question in machine learning has been whether an \say{ultimate} clustering algorithm is a possibility or merely a fool's errand.

Largely,
the question was addressed by \citet{wolpert1996lack} as a \term{No Free Lunch theorem}, a claim about the limitations of algorithms with respect to their problem domain. When an appropriate function is chosen to quantify the error (or \term{loss}), no algorithm can be superior to any other: an improvement across one subset of the problem domain is balanced by diminished performance on another subset. This is jarring at first. Are we not striving to find the best algorithms for our tasks? Yes---but by making specific assumptions about the subset of problems we expect to encounter, we can be comfortable tailoring our algorithms to those problems and sacrificing performance on remote cases.

\begin{figure}
\centering
	\begin{subfigure}[b]{0.48\linewidth}
	\includegraphics[width=\linewidth]{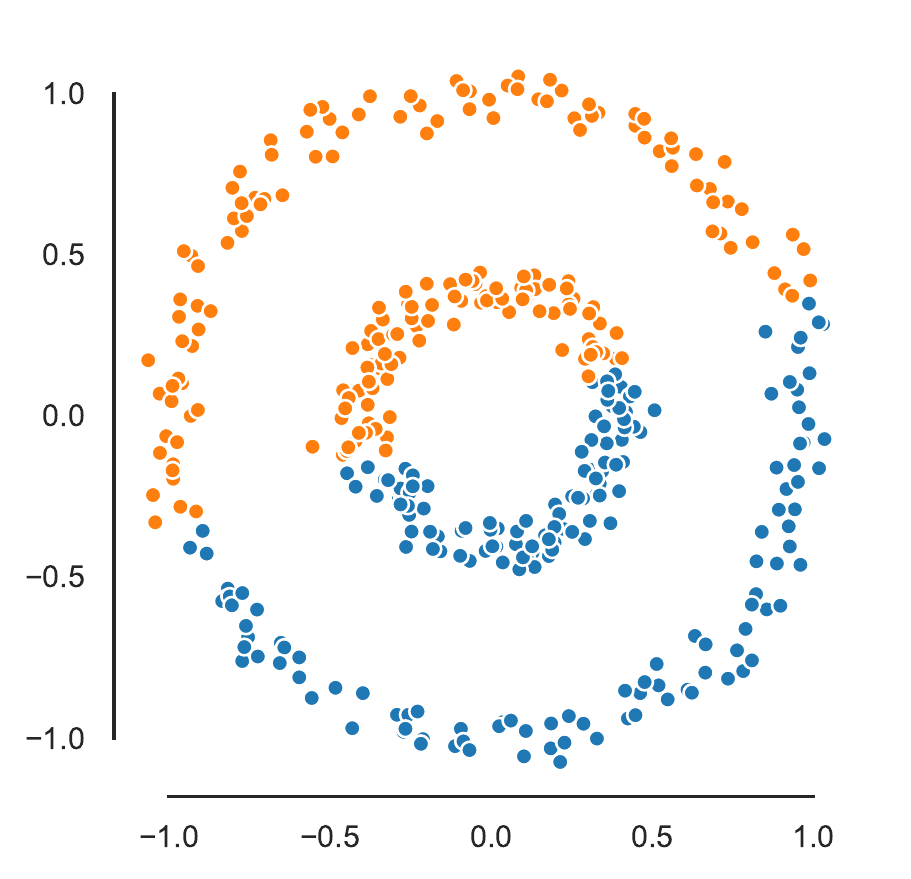}
	\caption{Non-spherical clusters}
	\end{subfigure}%
	~%
	\begin{subfigure}[b]{0.48\linewidth}
	\includegraphics[width=\linewidth]{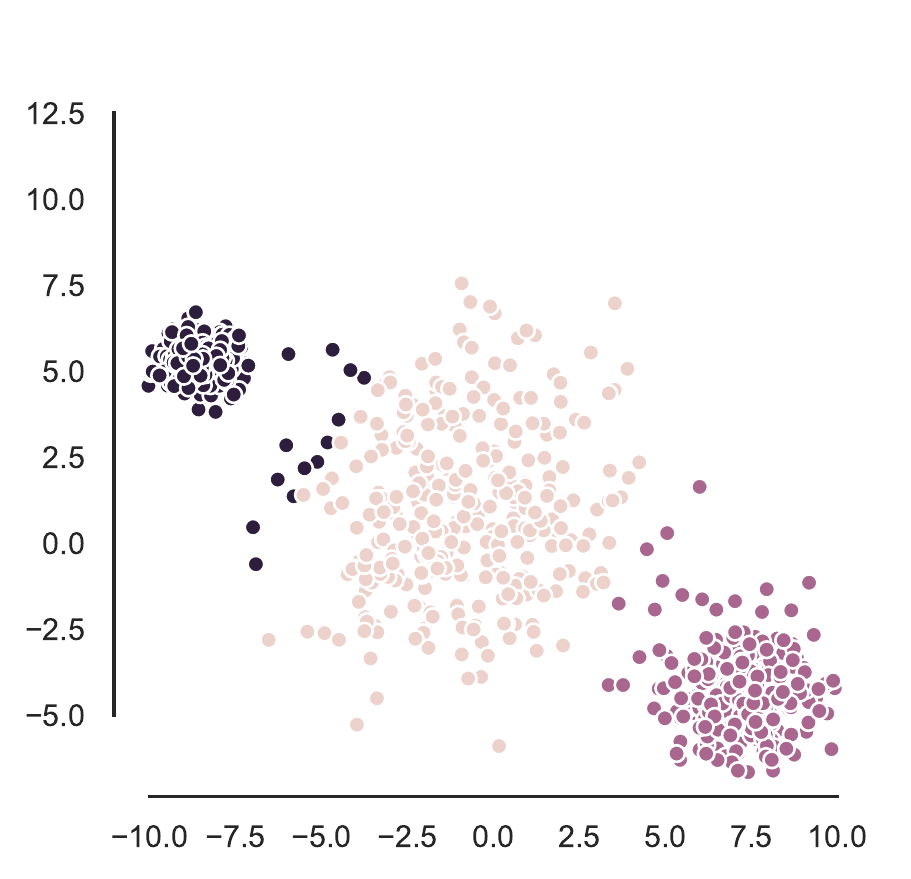}
	\caption{Unequal variances}
	\end{subfigure}

\caption{\(k\)-means clustering when certain assumptions are violated.}
\label{fig:means}
\end{figure}

As an example, the \(k\)-means algorithm for \(k\)-clustering is widely used for its simplicity and strength, but it assumes spherical clusters, equal variance in those clusters, and similar cluster sizes (equivalent to a homoscedastic Gaussian prior). \autoref{fig:means} shows the degraded performance on problems where these assumptions are violated.

To prove a No Free Lunch theorem for a particular task demands an appropriate loss function. A No Free Lunch theorem was argued for community detection \citep{peel2017ground}, using the adjusted mutual information function \citep{vinh2009information}.\footnote{Throughout this work, we assume that we evaluate against a known ground truth, as opposed to some intrinsic measure of partition properties like modularity \citep{newman2004modularity}.} However, the theorem is inexact. A No Free Lunch theorem relies on a loss function which imparts \term{generalizer-independence} (formally defined below): one which does not assume \emph{a priori} that some prediction is superior to another. The loss function used in the proof is only \emph{asymptotically} independent in the size of the input. We present a correction: by substituting an appropriate loss function, we are able to claim an exact version of the No Free Lunch theorem for community detection. The result generalizes to other set-partitioning tasks when evaluated with this loss function, including clustering, \(k\)-clustering, and graph partitioning.

\section{Background}

\subsection{Community detection}

A number of tasks on graphs seek a partition of the graph's nodes that maximizes a score function. Situated between the microscopic node-level and the macroscopic graph-level, these partitions form a \term{mesoscopic structure}---be it a core--periphery separation, a graph coloring, or our focus: \term{community detection} (CD). Community detection has been historically ill-defined \citep{radicchi2004defining, yang2016comparative}, though the intuition is to collect nodes with high interconnectivity
(or edge density) into communities with low edge density between them. The task is analogous to clustering, in which points near one another in a metric space are grouped together.

To assess whether the formulation of community detection matches one's needs, one performs extrinsic evaluation against a known \term{ground truth} clustering.
This ground truth can come from domain knowledge of real-world graphs or can be planted into a graph as a synthetic benchmark.
After running community detection on the graph, some similarity or error measure between the computed community structure and the correct one can be computed.

\paragraph{No bijection between true structure and graph}
Unfortunately, ground truth communities do not imply a single graph---and vice versa. \citet{peel2017ground} go as far as to claim, \say{Searching for the ground truth partition without knowing the exact generative mechanism is an impossible task.}

We can imagine the following steps for how problem instances are created, given that we have \(N = |V|\) nodes:
\begin{enumerate}
    \item Sample (true) partition \true{} $\in \Omega$;
    \item Generate graph $G$ from \true{} by adding edges according to the edge-generating process $g$.
\end{enumerate}
where \(\Omega\) is our \term{universe}: the space of all partitions of \(N = | V |\) objects. 
Given a graph $G = (V, E)$, we can imagine multiple truths \(\true_i \in \Omega\) that could define its edge set \(E\) by different generative processes~\(g_i: \Omega \to \Gamma \), where \(\Gamma\) is the set of all graphs with \(N\) nodes. \citet{peel2017ground} give a proof that extends from this simple example: Imagine that \(\true_1\) partitions the \(N\) nodes into \(N\) components (the \(N\)-partition), and \(\true_2\) partitions them into \(1\) component (the \(1\)-partition). Let \(g_1\) exactly specify the number of edges between each pair of communities, such that \(g_1(\true_1)\) is \(\graph\) with probability 1. Similarly, let \(g_2\) be an Erd\H{o}s--R\'enyi model such that \(g_2(\true_2)\) is \(\graph\) with nonzero probability. (\citet{peel2017ground} note that this is easily extended to graphs with more nodes.) We thus have two different ways to create a single graph; how can a method discern the correct one, without knowledge of \(g\)?

Community detection is then an ill-posed inverse problem: Use a function \(f: \Gamma \to \Omega\) to produce a clustering~\(\pred = f(G)\), which is hopefully representative of \true{} \citep[Appendix C]{peel2017ground}.\footnote{That is, the objective is to find \(f = g^{-1}\).} The function $f$ is not a bijection, so there isn't a unique \true{} represented in the given graph. Our algorithm~\(f\) must encode our prior beliefs about the generative process~\(g\) to select from among candidates. For this reason, we must hope that the benchmark graphs that we use are representative of the generative process for graphs in our real-world applications. That is, we hope that our benchmark domain matches our practical domain.

\paragraph{Other set-partitioning tasks}
While the remainder of this work focuses on community detection, our claims are relevant to other set-partitioning tasks.
Notable examples are clustering (the vector space analogue to community detection), graph $k$-partitioning, and \(k\)-clustering.
Metadata about the nodes and edges, such as vector coordinates, are used to guide the identification of such structure, but the tasks are all fundamentally set-partitioning problems. They can also have different universes \(\Omega\)---the latter tasks have a smaller universe than does community detection, for a given graph~\(G\): They consider only partitions with a fixed number of clusters.

\subsection{No Free Lunch theorems}

The \term{No Free Lunch theorem} in machine learning is a claim about the universal (in)effectiveness of learning algorithms. Every algorithm performs equally well when averaging over all possible input--output pairs. Formally, for any learning method $f$, the error (or \term{loss})~\loss{} of the method \(f\), summed over all possible problems \(( g, \true)\) equals a loss-specific constant \(\Lambda(\loss)\): 
\begin{equation}
\sum_{( g, \true )} \loss \left(\true, f\left(g\left(\true\right)\right)\right) = \Lambda(\loss)\text{,}
\label{eqn:nfl}
\end{equation}
defining the edge-generative process \(g\) and partition \(\true\) as above.
This loss is thus \emph{generalizer-independent}.
To reduce loss on a particular set of problems means sacrificing performance on others---\say{\emph{there is no free lunch}} \citep{wolpert1996lack, schumacher2001no}. Judiciously choosing which set to improve involves making assumptions about the distribution of the data: as we've mentioned, \(k\)-means is a method for \(k\)-clustering which works well on data with spherical covariance, similar cluster sizes, and roughly equal class sizes. When these assumptions are violated, performance suffers and overall balance is achieved.

\subsection{Community detection as supervised learning}

We follow \citet{peel2017ground} in framing the task of community detection (CD) as a learning problem. While recent algorithms, e.g.\ \citet{chen2018supervised}, have introduced learnable parameters to community detection algorithms the CD literature's algorithms are by and large untrained. These untrained algorithms encode knowledge of the problem domain in prior beliefs. We note that our work and \citet{peel2017ground} straightforwardly handle both of these cases.


In general supervised machine learning problems
, we seek to learn the function that maps an input space \(\mathbf{X}\) to an output space \(\mathbf{Y}\). We consider problem instances as sampled from random variables over each, so our goal is to learn the conditional distribution~\(p(Y \mid X)\). In the process of training on a dataset \(\mathcal{D}\), we develop a distribution over hypotheses~\(q\) which are estimates of the distribution~\(p\).

In the case of most community detection algorithms, our input space is the set of graphs on \(N\) nodes \(\Gamma\), and the output space is \(\omega\). There is no training data: \(\mathcal{D} = \varnothing\). All of our prior beliefs about \(p\) must be encoded in the prior distribution \(\Pr(q)\). That is, the model itself must contain our beliefs about the definition of community structure.
Only from the encoded \(\Pr(q)\) and an observed \(x \in \mathbf{X}\) (our graph \graph) do we form our point estimate of the true distribution~\(p\) \citep{peel2017ground}. However, in the case of trainable CD algorithms, we encode our beliefs in the posterior distribution~\(\Pr(q \mid \mathcal{D})\). 

\subsection{Loss functions and a priori superiority}

How should we evaluate an algorithm's predictions? Classification accuracy won't cut it: When comparing to the ground truth, there are no specific labels (e.g.\ no notion of a specific \say{Cluster 2})---only unlabeled groups of like entities. We settle for a measure of similarity in the groupings, quantifying how much the computed partition tells us about the ground truth.

A popular choice of measure is the \emph{normalized mutual information} \citep[NMI;][]{kvalseth1987entropy} between the prediction and the ground truth. While this measure has a long history in community detection, its flaws have been well-noted \citep{vinh2009information, peel2017ground,mccarthy2018normalized,mccarthy2019metrics}. It imposes a \say{geometric} structure upon the universe \(\Omega\),\footnote{To take the example of \citet{peel2017ground}, \(L^2\) loss (squared Euclidian distance) imposes a geometric structure: In the task of guessing points in the unit circle, guessing the center will garner a higher reward, on average, than any other point.} so something as simple as guessing the trivial all-singletons clustering outperforms methods that try at all to find a mesoscopic-level structure \citep{mccarthy2019metrics}. The property which NMI lacks is \emph{generalizer-independence}.

The property of generalizer-independence is defined by the generalization error function, an expectation of the loss $\Expect[L \mid p, q, \mathcal{D}]$. To satisfy this property, the generalization error must be independent of the particular true value \true{}. This is best expressed by \autoref{eqn:nfl}.

The adjusted mutual information (AMI, defined in \autoref{prevres}) \citep{vinh2009information} is a proposed replacement for NMI which does not impose a geometric structure upon the space. Unfortunately, this benefit is not fully realized when the expectation is computed over a space $\Psi \subset \Omega$. For the $\Psi$ used in \citet{peel2017ground}, the expected AMI across all problems is only \emph{asymptotically} generalizer-independent as the graph size grows---it is within some diminishing amount of error \(\varepsilon(N)\) of generalizer-independence, as proven by \citet{peel2017ground}.

\section{Previous Result: Approximate No Free Lunch Theorem}\label{prevres}

\citet{peel2017ground} frame community detection in the style of learning algorithms, letting them prove a No Free Lunch theorem for community detection. They note that the claim holds for \say{an appropriate choice of \ldots \(\mathcal{L}\)}---specifically a loss function~\(\mathcal{L}\) that is generalizer-independent---but their chosen loss function is not fully generalizer-independent. They also consider a stricter property than generalizer-independence:  \term{homogeneity}. With a homogeneous loss function, the \emph{distribution} of the error (not just its expectation) is identical, regardless of the ground truth. A measure which deviates from homogeneity may have this deviation bounded by a function of the number of vertices (the graph order).

\begin{lemma}[\citealp{peel2017ground}]
\label{thm:old-ami-homogeneous}
Adjusted mutual information (AMI) is a homogeneous loss function over the interior of the space of partitions of \(N\) objects, i.e., excluding the \(1\)-partition and the \(N\)-partition. Including these, AMI is homogeneous within~\(\frac{1}{\mathcal{B}_N}\).\footnote{ \(\mathcal{B}_N\) is the \(N\)-th Bell number, i.e., the number of partitions of a set of \(N\) nodes.}
\end{lemma}

\citet{wolpert1996lack} gives a generalized No Free Lunch theorem, which assumes a homogeneous loss.

\begin{theorem}[\citealp{wolpert1996lack}]
\label{thm:nfl-homogeneity}
For homogeneous loss \(\mathcal{L}\), the uniform average over all distributions \(p\) of \(\Pr\left(\ell \mid p, \mathcal{D}\right)\) equals \(\frac{\Lambda(\ell)}{|\mathbf{Y}|}\). \emph{(Plainly, \say{there is no free lunch}.)}
\end{theorem}

\citet{peel2017ground} then use Wolpert's result with their inexactly homogeneous measure to claim a No Free Lunch result.

\begin{theorem}[\citealp{peel2017ground}]
\label{thm:bad-nfl}
By \autoref{thm:old-ami-homogeneous} and \autoref{thm:nfl-homogeneity}, for the community detection problem with a loss function of AMI, the uniform average over all distributions \(p\) of \(\Pr(\ell \mid p, \mathcal{D})\) equals \(\frac{\Lambda(\ell)}{|\mathbf{Y}|}\).
\end{theorem}

But this choice of measure (AMI) is not, in fact, homogeneous over the \emph{entire} universe \(\Omega\) (\autoref{thm:old-ami-homogeneous}). A strategy that guesses either of the non-interior (i.e., boundary) partitions---the \(1\)-partition or \(N\)-partition---will yield a higher-than-average reward. There is indeed a negligible amount of free lunch---a free morsel, if you will.

\section{Diagnosis: Random Models}

\begin{figure}
\centering
	\begin{subfigure}[b]{0.48\linewidth}
	\includegraphics[width=\linewidth]{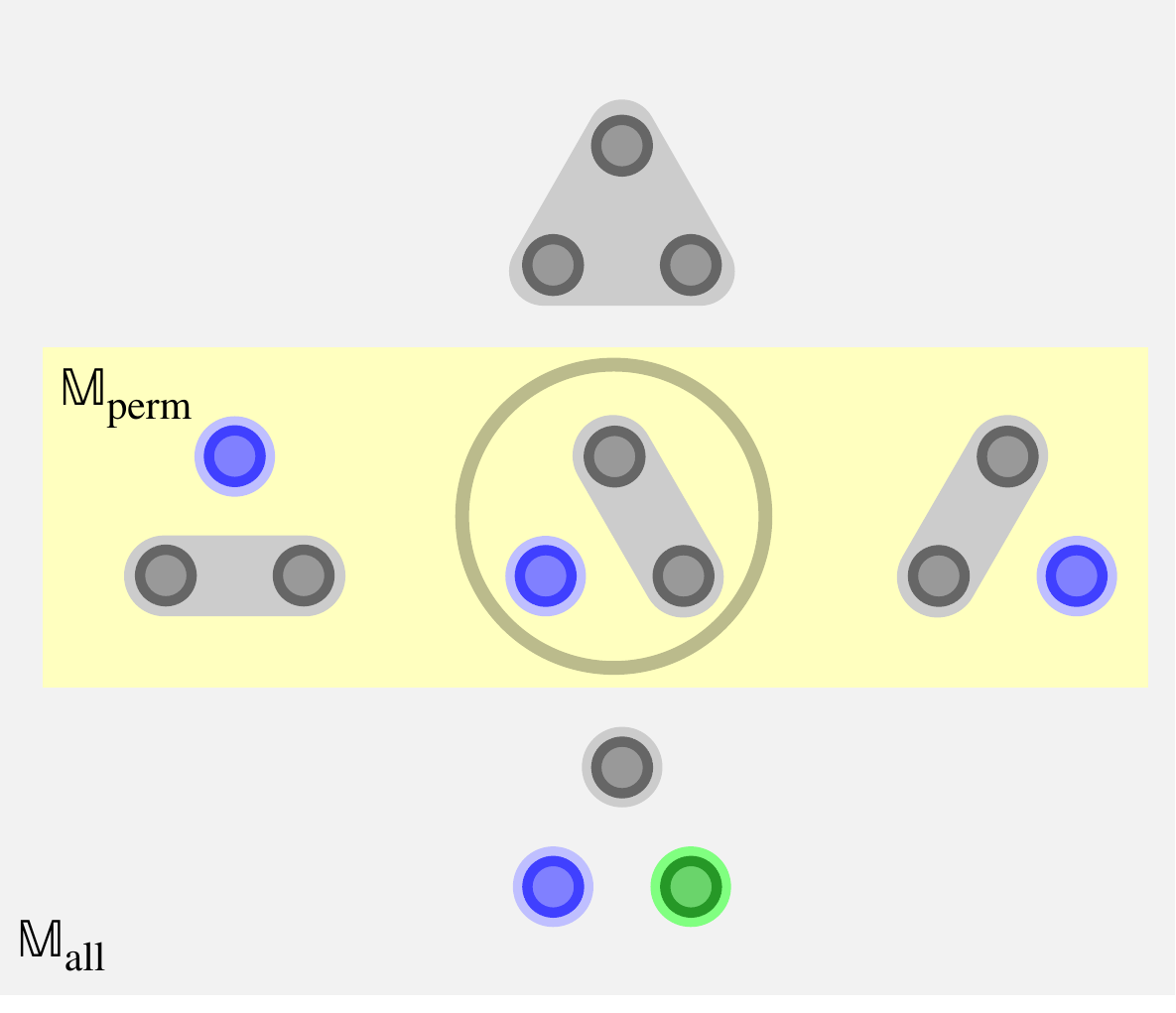}
	\caption{Ground truth has cluster size pattern \(\{2, 1\}\).}
	\end{subfigure}
	~
	\begin{subfigure}[b]{0.48\linewidth}
	\includegraphics[width=\linewidth]{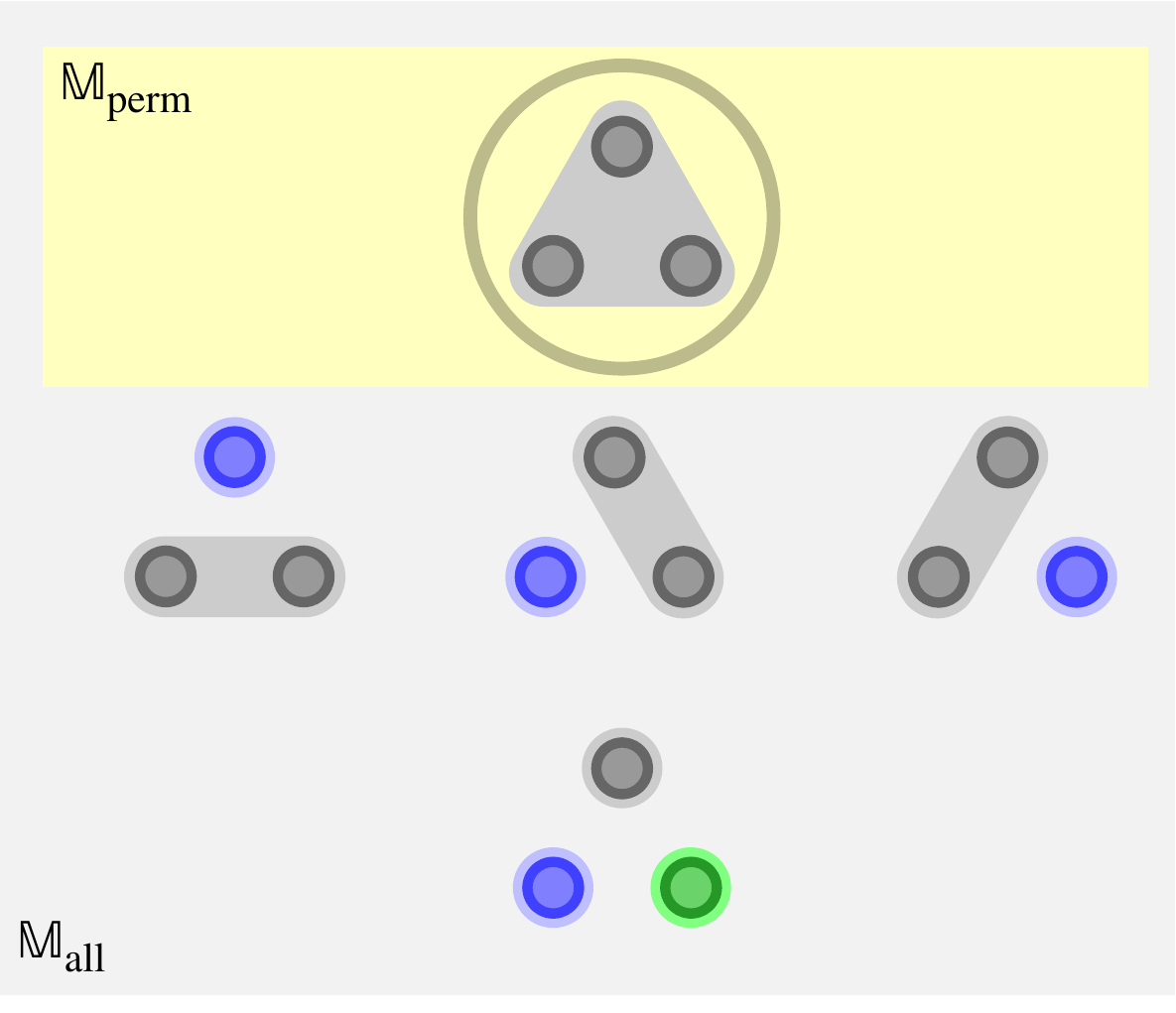}
	\caption{Ground truth has cluster size pattern \(\{3\}\).}
	\end{subfigure}

\caption{\model{all} and \model{perm} when clustering three nodes, for two different ground truths (circled). The top and bottom clusterings---the $1$ and $N$ clusterings---are the boundary partitions. All other partitions form the interior. \model{perm} changes based on the ground truth, but \model{all} stays the same.}
\label{fig:random-models}
\end{figure}

\citet{peel2017ground} use AMI out of the box, as proposed by \citet{vinh2009information}, which involves subtracting an expected value from a raw score. Unfortunately, AMI as given takes its expectation over the wrong distribution. 
Because of the mismatch, \citet{peel2017ground}'s claim of homogeneity is accurate only to within \(\frac{1}{\mathcal{B}_N}\) when considering the trivial partitions into either one community or \(N\)~communities.

Correcting this is arguably a pedantic demand, for two reasons:
\begin{enumerate}
\item The fraction \(\frac{1}{\mathcal{B}_N}\) converges to 0 superexponentially as \(N\) increases.
\item The deficiency is only present when \true{} is one of the trivial partitions. Otherwise, AMI as used is exactly homogeneous. But the trivial partitions reflect a lack of any mesoscopic community structure.
\end{enumerate}
Nevertheless, we'd like to see a tight claim of generalizer independence. To do this, we must select the proper \term{random model}, a sample space for a distribution. 

AMI adjusts NMI by subtracting the expected value from both the numerator and the denominator, shown in blue:
\begin{equation}
\label{eqn:ami}
\AMI(\pred, \true) \triangleq 
\frac%
	{I(\pred, \true) \mathcolor{blue}{- \Expect_{\pred', \true'}\left[ I(\pred', \true') \right] }}%
	{\mathcolor{magenta}{\max_{\pred', \true'} I(\pred', \true') } \mathcolor{blue}{- \Expect_{\pred', \true'}\left[ I(\pred', \true') \right] }}
\textrm{,}
\end{equation}
where \(I\) is the mutual information, maximized when the specific clustering~\pred{} equals the ground truth~\true{}. By inspecting \autoref{eqn:ami}, we see that AMI's value is \(1\) (the maximum) when \(\pred = \true\), \(0\) in expectation, and negative when the agreement between \pred{} and \true{} is worse than chance.

Subtly hidden in this equation is the decision of which distribution to compute the expectation over. For decades, this distribution has been what \citet{gates2017impact} call \model{perm}: all partitions of the same \term{partition shape}\footnote{A multiset of cluster sizes, also called a decomposition pattern \citep{hauer2016decoding} or a group-size distribution \citep{lai2016corrected}. It is equivalent to an integer partition of \(N\).} as \pred{} or \true{}. For example, if \pred{} partitioned 7 nodes into clusters of sizes 2, 2, and 3, then we would compute the expected mutual information over all clusterings where one had cluster sizes of 2, 2, and 3.

\citet{mccarthy2019metrics} argue that \model{perm} is inappropriate. To use this random model assumes that we can only produce outputs within that restricted space, when in actuality \(\Omega\) is the set of \emph{all} partitions of \(N\) nodes. Furthermore, during evaluation, we hold our ground truth fixed, rather than marginalizing over possible ground truths. Were we to instead consider a distribution over \true{}s, we would add noise from other possible generative processes which yield the same graph from different underlying partitions. In our average, we might be including scores on ground truths that better align with our notions of, say, core--periphery partitioning. For this reason, we take a \term{one-sided expectation}---over \(\mathcal{C}\), holding \(\mathcal{T}\) fixed. The one-sided distribution over all partitions of \(N\) nodes is called \modelone{all} \cite{gates2017impact}. This distribution is what we use for our AMI expectation, giving a measure denoted as \(\mathrm{AMI}_{\textrm{all}}^1\), which is recommended by \citet{mccarthy2019metrics}. It takes the form
\begin{equation}
\label{eqn:ami-all}
\AMI_{\mathrm{all}}^1(\pred, \true) \triangleq 
\frac%
	{I(\pred, \true) \mathcolor{blue}{- \Expect_{\pred^\prime \sim \modelone{all}}\left[ I(\pred^\prime, \true) \right] }}%
	{\mathcolor{magenta}{\max_{\pred'} I(\pred', \true) } \mathcolor{blue}{- \Expect_{\pred^\prime \sim \modelone{all}}\left[ I(\pred^\prime, \true) \right] }}
\textrm{.}
\end{equation}
The differences between \model{all} and \model{perm} are illustrated in \autoref{fig:random-models} under \(|V| = 3\).
We will now show that substituting \model{all} for \model{perm}, hence using \(\AMI_{\rm all}^1\), allows for an exact No Free Lunch theorem.

\section{An Exact No Free Lunch Theorem}

We strengthen the No Free Lunch theorem for community detection given by \citet{peel2017ground} by using an improved loss function, \(\mathrm{AMI}_{\textrm{all}}^1\), for community detection. Our proof does not distinguish the \say{boundary} partitions (the two trivial partitions) from the \say{interior} partitions (the remainder). It is entirely agnostic toward the particular ground truth \true{}, which is exactly what we need. 
We improve the previous result by moving from \model{interior} (which excludes the boundary partitions) to \model{all}.



\subsection{Generalizer-independence of \(\AMI_{\mathrm{all}}^1\)}

\begin{lemma}
\label{thm:homogeneity}
\(\AMI_{\mathrm{all}}^1\) is a generalizer-independent loss function over the \emph{entire} space \model{all} of partitions of \(N\)~objects.
\end{lemma}

\begin{proof}
Like \citet{peel2017ground}, we must show that the sum of scores is independent of~\true:
\begin{equation}
\label{eqn:L}
\forall \true_1, \true_2,\quad 
\sum_{\pred \in \Omega} \AMI_{\mathrm{all}}^1 \left(\pred, \true_1\right) 
=
\sum_{\pred \in \Omega} \AMI_{\mathrm{all}}^1 \left(\pred, \true_2\right)
\text{,}
\end{equation}
 where \(\Omega\)~is the space of all partitions of \(N\)~objects. Unlike \citet{peel2017ground}, we take the AMI expectation over all \(\mathcal{B}_N\) clusterings in \(\Omega\) using the random model \modelone{all} \citep{gates2017impact}. 

To prove our claim about \autoref{eqn:L}, we note that denominator of \(\AMI_{\mathrm{all}}^1\) is a constant with respect to \pred~(\autoref{eqn:ami-all}), so we can factor it out of the sum and restrict our attention to the numerator. This is because the max-term in the denominator is the constant \(\log N\) \citep{gates2017impact} and the expectation term for a given \true{} is independent of the particular \pred{}. Having factored this out, we will now prove \autoref{eqn:L} by the stronger claim:
\begin{equation}
\label{eqn:numerators}
\sum_{\pred \in \Omega} \left[ I(\pred, \true) - \Expect_{\pred' \sim \modelone{all}}\left[ I(\pred', \true) \right] \right] \stackrel{?}{=} 0 \quad \forall \, \true
\end{equation}
To prove \autoref{eqn:numerators}, we separate the summation's two terms:
\begin{equation}
\sum_{\pred \in \Omega} \left[ I(\pred, \true) \right] - \sum_{\pred \in \Omega} \left[ \Expect_{\pred' \sim \modelone{all}}\left[ I(\pred', \true) \right] \right]
\end{equation}
The expectation is uniform over the universe \(\Omega\),\footnote{Why do we assume uniformity over \(\Omega\)? Because this is the highest-entropy (i.e., least informed) distribution---it places the fewest assumptions on the distribution.} so we can apply the law of the unconscious statistician, then push the constant probability out, to get
\begin{equation}
\sum_{\pred \in \Omega} \left[ I(\pred, \true) \right] - \sum_{\pred \in \Omega} \left[ \frac{1}{|\Omega|} \sum_{\pred' \in \Omega}\left[ I(\pred', \true) \right] \right]
\end{equation}
Because 
the inner sum is independent of any particular \(\pred{}'\), the outer sum is a sum of constants---one for each element in \(\Omega\). We can now express \autoref{eqn:numerators} as follows, where the reciprocals straightforwardly cancel out:
\begin{equation}
\sum_{\pred \in \Omega} \left[ I(\pred, \true) \right] - |\Omega| \frac{1}{|\Omega|} \sum_{\pred' \in \Omega}\left[ I(\pred', \true) \right]  \equiv 0
\text{.}
\end{equation}
This equivalence implies that \autoref{eqn:L} is true.
\end{proof}

The proof is valid without loss of generality vis-\`a-vis the distribution---that is, as long as the AMI expectation is computed uniformly over the problem universe \(\Omega\), AMI is a generalizer-independent measure. This stipulation is relevant to tasks which assume a fixed number of clusterings---using \model{num}---like \(k\)-clustering and graph partitioning.  

Having demonstrated the generalizer-independence of AMI, we can define our loss function as, say, 
\begin{equation}
\loss (\pred, \true) = 1 - \AMI(\pred, \true)\text{.}
\end{equation}
 The loss is zero when we exactly match the true clustering and positive otherwise. 
 
Having proven the generalizer-independence of \(\mathrm{AMI}_{\mathrm{all}}^1\), we now turn to a more general form of the No Free Lunch theorem, which admits not just a homogeneous loss function but any generalizer-independent loss.

\begin{theorem}[\citealp{wolpert1996lack}]
\label{thm:nfl-gi}
For generalizer-independent loss \(\ell\), the uniform average over all \(p\), \(\mathbb{E}\left[\ell \;\middle\vert\; p, \mathcal{D}\right]\), equals \(\frac{\Lambda(\ell)}{|\mathbf{Y}|}\). (Plainly, \say{There is no free lunch.})
\end{theorem}
\begin{proof}
See \citet{wolpert1996lack}.
\end{proof}

\begin{theorem}[No Free Lunch theorem for community detection and other set-partitioning tasks]
For a set-partitioning problem with a loss function of adjusted mutual information \emph{using the appropriate random model for the task}, the uniform average over all \(p\), \(\mathbb{E}\left[\ell \;\middle\vert\; p, \mathcal{D}\right]\), equals \(\frac{\Lambda(\ell)}{|\mathbf{Y}|}\).
\end{theorem}
\begin{proof}
\autoref{thm:homogeneity} proves that AMI \emph{using the appropriate random model} is generalizer-independent. Applying \autoref{thm:nfl-gi} 
completes the proof \citep{peel2017ground}.
\end{proof}

\subsection{Other measures}

AMI stemmed from a series of efforts to improve normalized mutual information (NMI). We note that six other measures, when extended to \modelone{all} instead of \model{perm}, are also generalizer-independent: the adjusted Rand index \citep[ARI;][]{hubert1985comparing}, relative NMI \citep[rNMI;][]{zhang2015evaluating}, ratio of relative NMI \citep[rrNMI;][]{zhang2015relationship}, Cohen's \(\kappa\) \citep{liu2018evaluation}, corrected NMI \citep[cNMI][]{lai2016corrected}, and standardized mutual information \citep[SMI;][]{romano2014standardized}. We elide the proofs because they are similar to \autoref{thm:homogeneity}. Each of the six measures satisfies the precondition for the No Free Lunch theorem when the random model matches the problem domain. 

Of late, a renewed push has advocated using the adjusted Rand index \citep[ARI;][]{hubert1985comparing} to evaluate community detection; in fact, ARI and AMI are specializations of the same underlying function which uses \emph{generalized} information-theoretic measures~\citep{romano2016adjusting}. Every claim in the proof works for ARI, by replacing every mutual information \(I\) term with the Rand index \(\mathrm{RI}\).

Another line of research, focusing on improving NMI, produced rNMI \cite{zhang2015evaluating}, rrNMI \citep{zhang2015relationship}, and cNMI \citep{lai2016corrected}. We note that rrNMI is identical to one-sided AMI when both are extended to \modelone{all}. Consequently, our claim above works just as well for rrNMI. Further, because we were able to ignore the denominator of AMI in our proof of \autoref{thm:homogeneity}, we can do the same for rrNMI, which gives its unnormalized variant, rNMI. This means that rNMI is a generalizer-independent measure as well, when used in the appropriate one-sided random model. 
The practical benefit of normalizing rNMI into rrNMI is that the normalized measure gives a more interpretable notion of success. 

Additionally, \autoref{thm:homogeneity} holds true for standardized mutual information (which is equivalent to standardized variation of information and standardized V-measure) \citep{romano2014standardized}, the adjusted variation of information \citep{vinh2009information}, and for Cohen's $\kappa$, advocated for CD by \citet{liu2018evaluation}. This is because each measure shares the form of AMI: an observed score minus an expectation.

Finally, to show whether cNMI is generalizer-independent under the correct random model, we must show how to specialize it into a one-sided variant, because there is room for interpretation about how this should be done, even restricting our focus to \modelone{all}. The expression for cNMI
\begin{equation}
\cNMI(\pred, \true) \triangleq 
{
\frac%
	{2\NMI(\pred, \true) \mathcolor{blue}{- \Expect_{\pred'}\left[ \NMI(\pred', \true) \right] - \Expect_{\true'}\left[ \NMI(\pred, \true') \right] }}%
	{2 \mathcolor{blue}{- \Expect_{\pred'}\left[ \NMI(\pred', \pred) \right] - \Expect_{\true'}\left[ \NMI(\true, \true') \right] }}
}
\end{equation}
depends on both \pred{} and \true{} relative to the universes that contain them. Our specialization should remove dependence on the family of \true, so we arrive at the following expression after cancellation and noting that the NMI between a clustering and itself is 1:
\begin{equation}
\cNMI(\pred, \true) =
{
\frac%
	{\NMI(\pred, \true) \mathcolor{blue}{- \Expect_{\pred'}\left[ \NMI(\pred', \true) \right]}}%
	{1 \mathcolor{blue}{- \Expect_{\pred'}\left[ \NMI(\pred', \pred) \right]}}
}
\end{equation}
As it turns out, this quasi-adjusted measure is also generalizer-independent. 

In general, we now have a recipe for generalizer-independent loss functions: They can be created by subtracting the expected score from the observed score. This recipe works whenever a uniform expectation can be well defined.

\section{Conclusion}
We now have a proof of the No Free Lunch theorem for community detection and clustering that is both complete and exact.
We show that a corrected form of AMI, namely  \(\AMI_{\mathrm{all}}^1\), computes its expectation in a way that does not advantage the boundary partitions ($1$ cluster and $N$ singleton clusters). Indeed, this expectation is over the entire universe of partitions \(\Omega\), rather than any proper subset, such as the historically common \model{perm}.
%
We affirm the claim: \say{Any subset of problems for which an algorithm outperforms others is balanced by another subset for which the algorithm underperforms others. Thus, there is no single community detection algorithm that is best overall} \citep{peel2017ground}.

It is still possible for an algorithm to perform better on a \emph{subset} of community detection problems, so we can strive toward improved results on such a subset.
To hope to perform well, we must note the assumptions about the subset of problems we expect to encounter.
Some work has been done on estimating network properties to select the correct algorithm for the task at hand---a coarse way of checking assumptions \citep{peel2011estimating, yang2016comparative}. Beyond this, though, we must clarify what the problem of community detection \emph{is}; the formulation we choose will guide which subset of problem instances to prioritize and which to sacrifice.

\section*{Acknowledgments}
  The authors thank, alphabetically by surname, Daniel Larremore, Leto Peel, David Wolpert, Patrick Xia, and Jean-Gabriel Young for discussions that improved the work. Any mistakes are the authors' alone.

\bibliographystyle{acl_natbib}
\bibliography{nfl-thm}
\end{document}